\documentclass[11pt,fleqn]{article}
\usepackage{amssymb,latexsym,amsmath,amsfonts,graphicx, ebezier}

    \advance\textheight by \topskip

    \numberwithin{equation}{section}

    \def\Re{{\rm Re \,}}
    \def\Im{{\rm Im \,}}

    \def\bigO{{\cal O}}
    
    \def\Res{{\rm Res}}
    \def\sech{{\rm sech}}
    
    \newcommand{\ra}{\rightarrow}
    \def\P2n{{\rm P}_{{\rm II}}^{(n)}}

    \newtheorem{theorem}{Theorem}[section]

    \newtheorem{proposition}[theorem]{Proposition}
    
    \newtheorem{Definition}[theorem]{Definition}
    
    \newtheorem{Remark}[theorem]{Remark}
    \newenvironment{remark}{\begin{Remark}\rm}{\end{Remark}}
    \newtheorem{Example}[theorem]{Example}
    
    \newtheorem{Assumptions}[theorem]{Assumptions}

    \newcommand{\e}{\epsilon}
\newcommand{\lb}{\lambda}

    \newenvironment{proof}%
    {\rm \trivlist \item[\hskip \labelsep{\bf Proof. }]}%
    {\hspace*{\fill}$\Box$\endtrivlist}
    {\rm \trivlist \item[\hskip \labelsep{\bf Proof}]}%
    {\hspace*{\fill}$\Box$\endtrivlist}

\begin{document}
\title{Solitonic asymptotics for the Korteweg-de Vries equation in the small dispersion limit}
\author{T. Claeys and T. Grava}
\maketitle
\begin{abstract}
We study the small dispersion limit for the Korteweg-de Vries (KdV) equation $u_t+6uu_x+\epsilon^{2}u_{xxx}=0$ in a critical scaling regime where $x$ approaches the trailing edge of the region where the KdV solution shows oscillatory behavior. Using the Riemann-Hilbert approach, we obtain an asymptotic expansion for the KdV solution in a double scaling limit, which shows that the oscillations degenerate to sharp pulses near the trailing edge. Locally those pulses resemble soliton solutions of the KdV equation.
\end{abstract}

\section{Introduction}

We consider the Cauchy problem for the
Korteweg-de Vries (KdV) equation
\begin{equation}\label{KdV}
u_t+6uu_x+\epsilon^{2}u_{xxx}=0, \qquad u(x,t=0,\e)=u_0(x),\qquad \e>0, t>0,
\end{equation}
in the small dispersion limit where $\epsilon\to 0$.
We consider real analytic negative initial data with sufficient decay at infinity and with a single negative hump.
For small $t>0$, the solution to this problem can be approximated \cite{LL, DVZ} by the solution to the Cauchy problem for the (dispersionless) Hopf equation
\begin{equation}
u_t+6uu_x=0, \qquad u(x,0)=u_0(x),\qquad t>0,
\end{equation}
which can be solved using the method of characteristics.
At time
\[
t_c=\dfrac{1}{\max_{\xi\in\mathbb{R}}[-6u'_0(\xi)]},
\]
the Hopf equation reaches a point of gradient catastrophe where the derivative of the Hopf solution blows up.
For $t>t_c$, the Hopf solution only has a multi-valued continuation.
However, the KdV solution is well-defined for all positive $t$ and $\e>0$: the dispersive term $\epsilon^{2}u_{xxx}$ regularizes the gradient catastrophe.
For $t$ slightly bigger than $t_c$, the KdV solution $u(x,t,\epsilon)$ develops rapid oscillations \cite{LL, FRT} that in the limit $\epsilon\rightarrow 0$ are confined in a certain interval $[x^-(t),x^+(t)]$, see Figure \ref{fig1}.
\begin{figure}[!htb]\label{fig1}
\begin{center}
\includegraphics[width=5in]{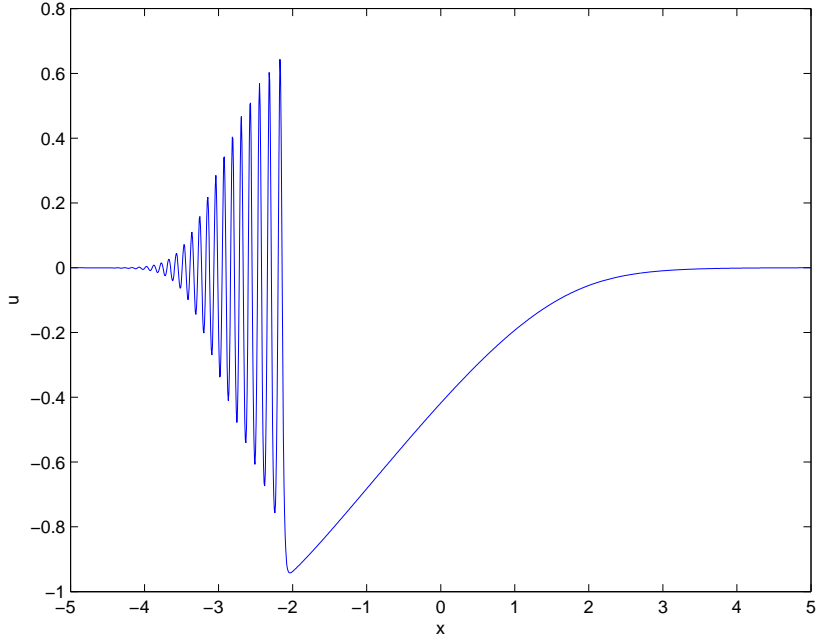}
\end{center}
\caption{Solution of the KdV equation with $\e=10^{-2}$, initial data $u_0(x)=-\mbox{sech}^2(x)$, and $t=0.4$. }
\end{figure}
In the $(x,t)$-plane, the oscillations take place in a cusp-shaped region (which depends on the initial data) as illustrated in Figure 2.
\begin{figure}[t]\label{figcusp}
\begin{center}
\includegraphics[scale=0.35]{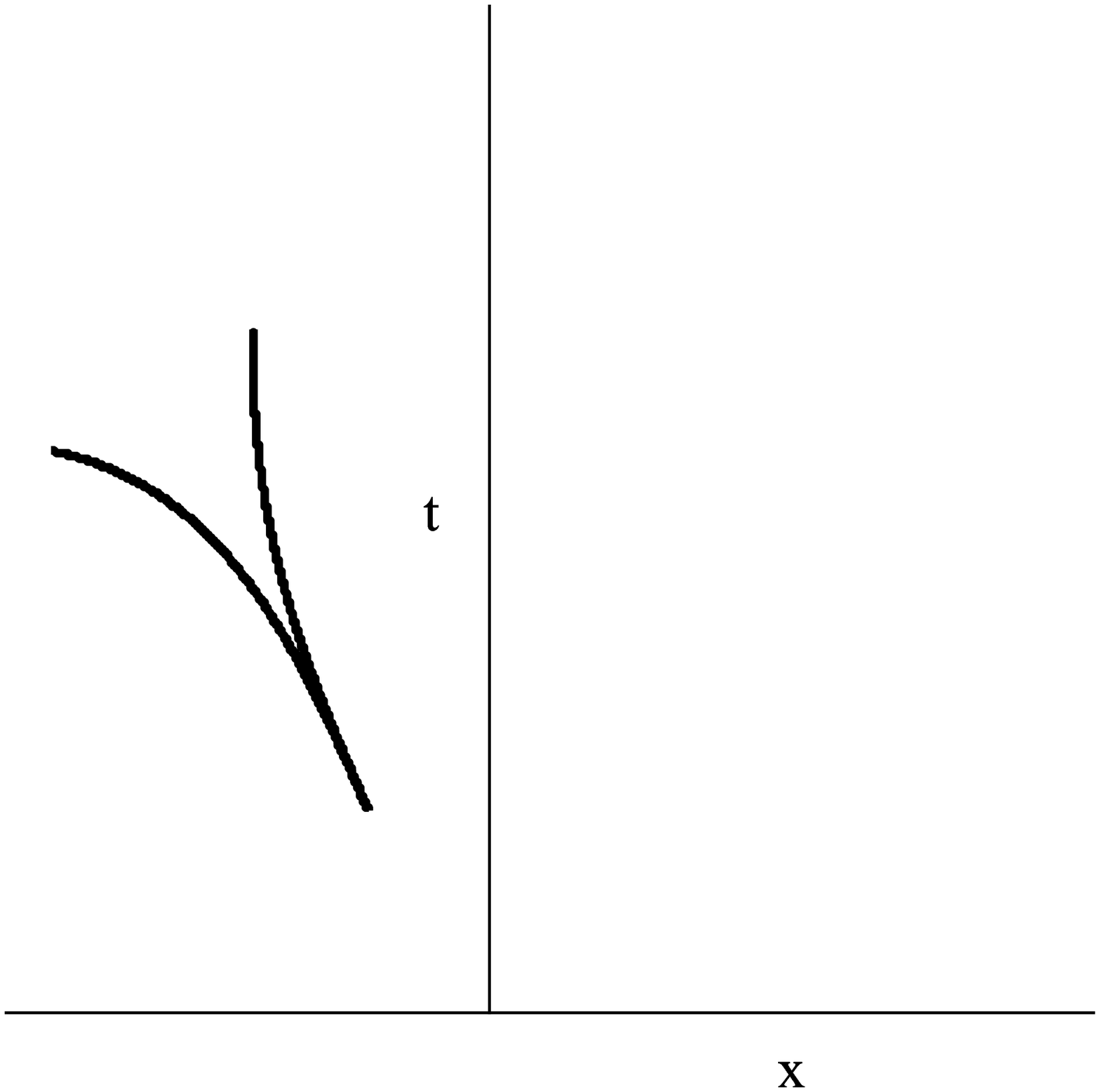}
\end{center}
\caption{Cusp-shaped oscillatory region in the $(x,t)$-plane for $t>t_c$.}
\end{figure}
Inside the cusp-shaped region for $t$ slightly bigger than $t_c$, the KdV solution $u(x,t,\epsilon)$ can be written, asymptotically for small $\epsilon$, in terms of the Jacobi elliptic theta function and the complete elliptic integrals of the first and second kind $E(s)$ and $K(s)$ \cite{LL, DVZ, V2, GP}:
\begin{equation}
\label{elliptic} u(x,t,\e)\simeq \beta_1+\beta_2+\beta_3+2\alpha+
2\e^2\frac{\partial^2}{\partial
x^2}\ln\vartheta(\Omega(x,t);\mathcal{T}),
\end{equation}
where $\Omega$, $\alpha$, and $\mathcal{T}$ have the form
\begin{align}
&\label{Omega}
\Omega(x,t)=\dfrac{\sqrt{\beta_1-\beta_3}}{2\e K(s)}[x-2 t(\beta_1+\beta_2+\beta_3) -q],
\\
&\label{alpha}
\alpha(s)=-\beta_{1}+(\beta_{1}-\beta_{3})\frac{E(s)}{K(s)},\;\;\mathcal{T}=i\dfrac{K'(s)}{K(s)},
\;\; s^{2}=\frac{\beta_{2}-\beta_{3}}{\beta_{1}-\beta_{3}}.
\end{align}
Note that $K'(s)=K(\sqrt{1-s^{2}})$, and
 $\vartheta$ is defined by the
Fourier series
\[
\vartheta(z;\mathcal{T})=\sum_{n\in\mathbb{Z}}e^{\pi i n^2\mathcal{T}+2\pi i nz}.
\]
The formula for $q$ in the phase $\Omega$ in (\ref{Omega}) is equal to \cite{GK2, DVZ}
\begin{equation}
\label{q0}
    q(\beta_{1},\beta_{2},\beta_{3}) = \frac{1}{2\sqrt{2}\pi}
    \int_{-1}^{1}\int_{-1}^{1}d\mu d\nu \frac {f_L( \frac{1+\mu}{2}(\frac{1+\nu}{2}\beta_{1}
    +\frac{1-\nu}{2}\beta_{2})+\frac{1-\mu}{2}\beta_{3})}{\sqrt{1-\mu}
    \sqrt{1-\nu^{2}}},
\end{equation}
where $f_L(y)$ is the inverse function of the decreasing part of the initial data $u_0$. Formula (\ref{elliptic}) can be written also in terms of Jacobi elliptic function $\mbox{dn}$  in the form
\begin{equation}
\label{elliptic1}
\begin{split}
u(x,t,\e)&\simeq \beta_2+\beta_3-\beta_1+2\dfrac{\beta_1-\beta_2}{\mbox{dn}^2(2K(s)\Omega)}\\
&\simeq
\beta_2+\beta_3-\beta_1+2(\beta_1-\beta_3)\mbox{dn}^2(2K(s)\Omega+(2k+1)K(s)),
\end{split}
\end{equation}
for any integer $k$ since $\mbox{dn}(u;s)$ is periodic with period
$2K(s)$.

For constant values of the $\beta_i$'s, the right hand side of (\ref{elliptic1}) is an exact solution of KdV.
However in  the description
of the leading order asymptotics of $u(x,t,\e)$ as $\e\ra 0$,
 the numbers $\beta_1>\beta_2>\beta_3$ depend on $x$ and $t$  and evolve
 according to the Whitham equations \cite{W}
\begin{equation}
\label{Whitham}
\dfrac{\partial}{\partial t}\beta_i+v_i\dfrac{\partial}{\partial x}\beta_i=0,\quad   v_{i}=4\frac{\prod_{k\neq
     i}^{}(\beta_{i}-\beta_{k})}{\beta_{i}+\alpha}+2(\beta_1+\beta_{2}+\beta_{3}),\;\;\; i=1,2,3,
\end{equation}
with $\alpha$ as in (\ref{alpha}).

The Whitham equations  (\ref{Whitham}) can be integrated through the
so-called hodograph transform, which generalizes the method of characteristics,
and which gives the solution in the implicit form \cite{T}
\begin{equation}
\label{hodograph}
x=v_it+w_i,\quad i=1,2,3,
\end{equation}
where  the $v_i$'s  are defined in (\ref{Whitham}) and $w_i=w_i(\beta_1,\beta_2,\beta_3)$ for $i=1,2,3$
is obtained from an algebro-geometric procedure by the formula \cite{FRT}
\begin{equation}
    w_{i} =
    \frac{1}{2}\left(v_{i}-2\sum_{k=1}^{3}\beta_{k}\right)\frac{
    \partial q}{\partial\beta_{i}}+q,\quad i=1,2,3,
    \label{eq:w}
\end{equation}
with $q$ defined in (\ref{q0}).
 The formula (\ref{q0}) for $q$ is valid as long as $\beta_3$ does not reach the minimal value of the initial data $u_0$. When $\beta_3$ reaches the negative hump it is also necessary to take into account the increasing part of the initial data $f_R$. However,  formula (\ref{q0}) is sufficient for the purpose of this manuscript.

\medskip

Near the boundary of the oscillatory cusp-shaped region, neither the Hopf asymptotics nor the elliptic asymptotics are satisfactory. Three different transitional regimes can be distinguished: (1) the cusp point where the gradient catastrophe for the Hopf equation takes place and where $\beta_1=\beta_2=\beta_3$, (2) the leading edge of the oscillatory zone where $\beta_2=\beta_3$, and (3) the trailing edge of the oscillatory zone where $\beta_1=\beta_2$.

Near the point of gradient catastrophe, it was conjectured \cite{Dubrovin, Dubrovin1} and proved afterwards \cite{CG} that the asymptotics for the KdV solution are given in terms of
a distinguished Painlev\'e transcendent, namely a special smooth solution $U(X,T)$ to the fourth order ODE
\[
X=T\, U -\left[ \dfrac{1}{6}U^3  +\dfrac{1}{24}( U_{X}^2 + 2 U\, U_{XX} )
+\frac1{240} U_{XXXX}\right],
\]
which is the second member of the first Painlev\'e hierarchy. In a double scaling limit where $\e\to 0$ and simultaneously $x$ and $t$ approach the point and time of gradient catastrophe $x_c$ and $t_c$ at an appropriate rate, the KdV solution has an expansion of the following form \cite{CG},
\[
u(x,t,\e)=u_c +\left(\dfrac{2\epsilon^2}{k^2}\right)^{1/7}
U \left(
\dfrac{x- x_c-6u_c (t-t_c)}{(8k\epsilon^6)^{\frac{1}{7}}},
\dfrac{6(t-t_c)}{(4k^3\epsilon^4)^{\frac{1}{7}}}\right) +\bigO\left( \epsilon^{4/7}\right),
\]
where $k$ is a constant depending on the initial data. This expansion holds for negative real analytic initial data with a single negative hump and sufficient decay at infinity, but it is conjectured by Dubrovin that it has a universal nature and extends to all Hamiltonian perturbations of hyperbolic systems near the point of gradient catastrophe. It is remarkable that the function $U(X,T)$ itself is a solution to the (re-scaled) KdV equation
$U_T+UU_X+\frac{1}{12}U_{XXX}=0$.

Near the leading edge at the left of the zone where the oscillations appear (see Figure \ref{fig1}), numerical results \cite{GK2} showed that the amplitude of the oscillations is asymptotically described by the Hasting-McLeod solution to the second Painlev\'e equation. For $t$ slightly bigger than $t_c$, in a double scaling limit where $\e\to 0$ and simultaneously $x$ approaches the leading edge $x^-$ at an appropriate rate, it was proved in \cite{CG2} that
\[
\label{expansionu}
u(x,t,\e) = u-\dfrac{4\e^{1/3}}{c^{1/3}}A\left[-\frac{x-x^-}{c^{1/3}\sqrt{u-v}\,\e^{2/3}}\right]
\cos\left(\frac{\Theta(x,t)}{\epsilon}\right)
+\bigO(\e^{2/3}),
\]
with $A$  the Hasting-McLeod solution to the second Painlev\'e equation,  the phase
\[\Theta(x,t)=2\sqrt{u-v}(x-x^-)+2\int_{v}^{u}(f_L'(\xi)+6t)\sqrt{\xi-v}d\xi,\]
and $x^-(t)$, $u(t)$ and $v(t)$ are obtained from (\ref{hodograph}) in the limit $\beta_2=\beta_3=v$.
This expansion holds, just like near the gradient catastrophe, for negative real analytic initial data with a single negative hump and sufficient decay at infinity. It is natural to expect that it is, to a certain extent, universal for a whole class of equations.

Our aim is to give an asymptotic description of the KdV solution
near the trailing edge. Here the behavior of the KdV solutions is
genuinely different compared to the Painlev\'e type behaviors near
the point of gradient catastrophe and the leading edge. This is not
surprising when considering Figure \ref{fig1}, where we observe that
the amplitude of the oscillations is of order $\bigO(1)$ near  the
trailing edge, whereas it smoothly decays towards the leading edge.
We will show that, in the limit where $\e\to 0$, the last
oscillations at the right end of the oscillatory zone behave like
solitons which are, in the local scale, at a large distance away
from each other. Recall that the KdV equation admits soliton
solutions of the form $a\,\sech^2(bx-ct)$. We should note that the
trailing edge asymptotics we will obtain show remarkable
similarities with recently obtained critical asymptotics for the
focusing nonlinear Schr\"odinger equation \cite{BT}.

The trailing edge $x^+(t)$ of the oscillatory interval (i.e.\ the right edge of the cusp-shaped region in Figure 2) is uniquely determined by the equations \cite{FRT, GT}
\begin{align}
&\label{trailing1}x^+(t)=6tu(t)+f_L(u(t)),\\
&\label{trailing2}6t+\theta(v(t);u(t))=0,\\
&\label{trailing3}\int_{u(t)}^{v(t)}(6t+\theta(\lambda;u(t)))\sqrt{\lambda-u(t)}d\lambda=0.
\end{align}
Here $v(t)>u(t)$ and
\begin{equation}
\label{theta}
\theta(v;u)=\dfrac{1}{2\sqrt{v-u}}\int_{u}^v\dfrac{f'_L(\xi)d\xi}{\sqrt{v-\xi}},
\end{equation}
and $f_L$ is the inverse of the decreasing part of $u_0$. System
(\ref{trailing1})-(\ref{trailing3}) is the confluent form of the
Whitham equations where $\beta_1=\beta_2$.

\medskip

Before stating our result, let us take a look at the elliptic
asymptotics for KdV in the limit $\beta_1\rightarrow\beta_2$. Note
that this is a formal limit of the asymptotic expansion for KdV and
that there is no rigorous argument to justify that this actually
leads to a good approximation for the KdV solution near the trailing
edge. The phase defined in (\ref{Omega}) behaves like
\[
2K(s)\Omega\simeq\dfrac{x-x^+}{\epsilon}\sqrt{v-u},\qquad \mbox{as
$\beta_1-\beta_2\to 0$},
\]
and
\[
K(s)\simeq\dfrac{1}{2}\ln\left[\dfrac{8}{1-s}\right],\quad
\mbox{as}\;\;s\rightarrow 1.
\]
In addition the Jacobi elliptic function $\mbox{dn}(u;s)\rightarrow
\sech(u)$ as the modulus $s\rightarrow 1$. Therefore the formal
limit of the elliptic solution (\ref{elliptic1}) as $s\rightarrow
1$, $\beta_1,\beta_2\rightarrow v$, $\beta_3\rightarrow u$, gives
\begin{equation}\label{expansion formal}
u(x,t,\epsilon)\simeq
u+2(v-u)\,\sech^2\left[\dfrac{x-x^+}{\epsilon}\sqrt{v-u}+(k+\dfrac{1}{2})\ln\left[\dfrac{8}{1-s}\right]\right]
.
\end{equation}
We would like to underline that the limit we computed in the above
expression has appeared many times in the literature, starting from
the seminal paper of Gurevich-Pitaevski \cite{GP}  but the important
second term of the phase in (\ref{expansion formal}) which was
calculated later  by Deift, Venakides and Zhou \cite{DVZ}, was  and
is still ignored in some literature. When taking the limit $s\to 1$
for fixed $\epsilon$, the solitonic $\sech^2$-term in
(\ref{expansion formal}) is of order $\bigO(\beta_1-\beta_2)$.
However it is clear from Figure \ref{fig1} that the amplitude of the
oscillations near $x^+$ is  of order $\bigO(1)$. This indicates that
it is necessary to perform a double scaling limit letting
$\epsilon\rightarrow 0$ and $x\rightarrow x^+$ simultaneously at an
appropriate rate. If we want to capture the top of an oscillation
close to the trailing edge, we should take the double scaling limit
in such a way that the phase of the $\sech^2$ term is zero for some
integer $k$. If $k\geq 0$, this will turn out to be consistent with
the rigorous asymptotic expansion for $u$ we will obtain below in
Theorem \ref{main theorem}, up to a phase shift.

\medskip

Similarly as in \cite{CG, CG2}, we consider initial data $u_0(x)$ which satisfy the following conditions.
\begin{Assumptions}\label{assumptions}\
\begin{itemize}
\item[(a)] $u_0(x)$ is real analytic and has an analytic continuation to the complex plane in a domain of the form
\[
\mathcal{S}=\{z\in\mathbb{C}: |\Im z|<\tan\theta_0|\Re z|\}\cup\{z\in\mathbb C: |\Im z|<\sigma\},
\]
for some $0<\theta_0<\pi/2$ and $\sigma>0$,
\item[(b)] $u_0(x)$ decays as $x\rightarrow\infty$ in $\mathcal{S}$ such that
\begin{equation}
u_0(x)=\bigO\left(\frac{1}{|x|^{3+s}}\right), \;\;s>0,
\end{equation}
\item[(c)] for real $x$,  $u_0(x)$ is negative and it has a single local minimum at a certain point
$x_M$, with
\[u_0'(x_M)=0, \qquad u_0''(x_M)>0,\] for simplicity we assume that $u_0$ is normalized such
that $u_0(x_M)=-1$,
\item[(d)] the point of gradient catastrophe for the Hopf equation is 'generic' in the sense that
\begin{equation}
\label{genericity}
f_L'''(u_c)\neq 0,
\end{equation}
where $f_L$ is the inverse of the decreasing part of the initial data $u_0$.
\end{itemize}
\end{Assumptions}
Because of condition (b), we will be able to apply direct and inverse scattering theory, and to setup a Riemann-Hilbert (RH) problem for the Cauchy problem of KdV. Conditions (a) and (c) will enable us to keep control over the reflection coefficient in the small dispersion limit; condition (d) will be needed only in the core of the RH analysis, but is nevertheless needed for the result stated below.

\begin{theorem}\label{main theorem}
Let $u_0(x)$ be initial data for the Cauchy problem of the KdV equation
satisfying Assumptions \ref{assumptions}, and let $x^+=x^+(t)$, $u=u(t)$, and $v=v(t)$ solve the system (\ref{trailing1})-(\ref{trailing3}). There exists $T>t_c$ such that for $t_c<t<T$, we have the following expansion for the KdV solution $u(x,t,\e)$ as $\e\to 0$,
\begin{equation}\label{expansion u}
u\left(x^++\frac{\e\ln\e}{2\sqrt{v-u}}y,t,\e\right)=u+2(v-u)\sum_{k=0}^{\infty}\sech^2(X_k)+\bigO(\e\ln^2\e),
\end{equation}
where
\begin{equation}
\begin{split}\label{Xk}
&X_k=\frac{1}{2}(\frac12-y+k)\ln\e-\ln(\sqrt{2\pi} h_k)-(k+\frac12)\ln\gamma,\\
&h_k=\dfrac{2^{\frac{k}{2}} }{\pi^{\frac{1}{4}}\sqrt{k!}},\quad \gamma=4(v-u)^{\frac{5}{4}}\sqrt{-\partial_v\theta(v;u)},
\end{split}
\end{equation}
and $\theta$ is given by (\ref{theta}). This expansion holds point-wise for $y\in\mathbb R$ and uniformly for $y$ bounded.
\end{theorem}
\begin{figure}[t]\label{fig cartoon}
\begin{center}
\includegraphics[width=3in, height=2in]{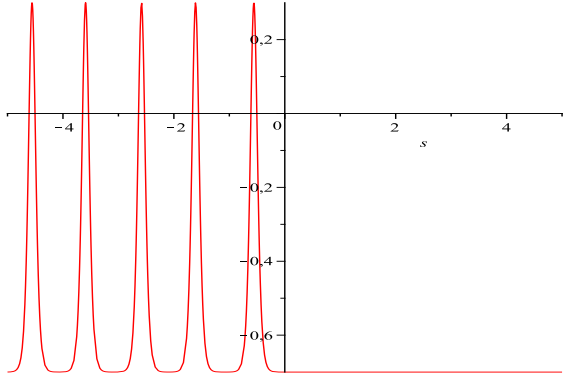}
\end{center}
\caption{Cartoon picture of
$u\left(x^++\frac{\e\ln\e}{2\sqrt{v-u}}y,t,\e\right)$ as a function
of $s=-y$ for $\e=10^{-5}$, or in other words a magnified version of
Figure \ref{fig1} near the trailing edge with the $x$-axis stretched
with a factor $\bigO(1/(\e|\ln\e|))$. One observes narrow pulses
near the negative half integers.}
\end{figure}
\begin{remark}
Observe that each term in the sum of (\ref{expansion u}) generates a pulse with amplitude $2(v-u)$ for $y$ near a half positive integer (as shown in Figure 3), which can be seen as a soliton. The term $\sech^2(X_k)$ is centered at $y\approx k+\frac{1}{2}$, and it already decreased to order $\bigO(\epsilon^{\frac{1}{2}})$ for $y$ near $k$ and $k+1$. Near $k-\frac{1}{2}$ and $k+\frac{3}{2}$, the contribution of $\sech^2(X_k)$ is absorbed by the error term $\bigO(\e\ln^2\e)$.
For any $y$, the infinite sum in (\ref{expansion u}) thus reduces to the sum of the two solitons centered closest to $y$. The sum of all the other terms is of order $\bigO(\e\ln^2\e)$.
Values of $y$ for which $y\leq -\frac{1}{2}$ lead us out of the oscillatory zone: the contribution of the sum of solitons to (\ref{expansion u}) is small, in this case we have $u(x,t,\e)=u(t)+\bigO(\e\ln^2\e)$.
\end{remark}
\begin{remark}
Whereas the function $u$ is clearly decreasing for $x$ to the right of the trailing edge in Figure \ref{fig1}, this is not visible in Figure 3 or in the asymptotic expansion (\ref{expansion u}). The reason is that this effect is of order $\e\ln\e$ in the local variable $y$. On the other hand, it is not clearly visible in Figure \ref{fig1} that the oscillations degenerate to sharp pulses towards the trailing edge. This is a consequence of the fact that the length of the pulses scales with $\bigO(1/|\ln\e|)$, which decreases slowly as $\e\to 0$.
\end{remark}
\begin{remark}
Analyzing the Whitham equations carefully in the limit
$\beta_1\to\beta_2$, it is possible to deduce an expansion of the
form (\ref{expansion u}) from (\ref{expansion formal}), but the
formally obtained $X_k$ misses a constant term (depending on $k$)
compared to (\ref{Xk}).
\end{remark}
\begin{remark}
The time $T>t_c$ appearing in Theorem \ref{main theorem} should be sufficiently small such that Proposition \ref{prop phi} holds. Loosely speaking, this means that the $G$-function which we will construct below is not allowed to have any singular behavior for $t_c<t<T$.
\end{remark}
\begin{remark}The techniques we will use to prove Theorem \ref{main theorem} are very close to the ones used in \cite{BertolaLee, C, Mo}, where orthogonal polynomials on the real line with respect to a critical exponential weight were studied. Those polynomials describe the birth of a cut in unitary random matrix ensembles, see also \cite{Eynard}. Although the focus in \cite{BertolaLee, C, Mo} was on the random matrix eigenvalues rather than on the orthogonal polynomials, we expect that the recurrence coefficients for the associated orthogonal polynomials admit an asymptotic expansion similar to (\ref{expansion u}) in an appropriate double scaling limit.
\end{remark}

\medskip

In Section \ref{section 2}, we will set up the RH problem for the KdV equation, and we will recall asymptotic results
 about the reflection coefficient in the small dispersion limit.
 In Section \ref{section 3} we will analyze asymptotically the RH problem using the Deift/Zhou steepest descent method.
 At the heart of the analysis lies the construction of a local parametrix built out of Hermite polynomials.
 The degree of the Hermite polynomials depends on the scaling variable $y$.
 The transitions where the degree of the Hermite polynomials increases takes place for $y$ near positive half integers
 and are responsible for the presence of the spikes in the asymptotic behavior of the KdV solution $u$.

\section{RH problem for KdV and reflection coefficient}\label{section 2}

It is well-known that solutions to the KdV equation can be expressed in terms of a RH problem. This fact relies on the direct and inverse scattering transform \cite{BDT, Faddeev}. Consider the following RH problem.
\subsection{RH problem for $M$}
\begin{itemize}
\item[(a)] $M(\lb;x,t,\e):\mathbb C\backslash \mathbb{R}\to \mathbb C^{2\times 2}$ is analytic. \item[(b)] $M$ has continuous boundary values $M_+(\lb)$ and $M_-(\lb)$ when approaching $\lambda\in\mathbb R\setminus\{0\}$ from above and below, and
\begin{align*}&M_+(\lb)=M_-(\lb){\small \begin{pmatrix}1&r(\lb;\e)
e^{2i\alpha(\lambda;x,t)/\e}\\
-\bar r(\lb;\e)e^{-2i\alpha(\lb;x,t)/e}&1-|r(\lb;\e)|^2
\end{pmatrix}},&\mbox{ for $\lb<0$,}\\
&M_+(\lb)=M_-(\lb)\sigma_1,\quad
\sigma_1=\begin{pmatrix}0&1\\1&0\end{pmatrix},&\mbox{ for $\lb>0$},
\end{align*}
with $\alpha$ given by
\begin{equation}\label{def alpha1}\alpha(\lambda;x,t)=4t(-\lambda)^{3/2}+x(-\lambda)^{1/2},\end{equation} defined with the branches of $(-\lambda)^{\theta}$ which are analytic in $\mathbb C\setminus [0,+\infty)$ and positive for $\lambda<0$.
Furthermore for fixed $x,t,\e$, we impose $M(\lb)$ to be bounded near $0$.
\item[(c)] $M(\lb;x,t,\e)=\left(I+\bigO(\lambda^{-1})\right)
\begin{pmatrix}1&1\\&\\i\sqrt{-\lb}&-i\sqrt{-\lb}\end{pmatrix}\;\;$
as $\lb\rightarrow \infty$.
\end{itemize}
In particular we can write
\[M_{11}(\lb;x,t,\e)=1+\dfrac{M_{1, 11}(x,t,\e)}{\sqrt{-\lb}}+\bigO(1/\lb),\qquad\mbox{ as $\lb\to\infty$}.\]
In general it is true that, if the RH problem for $M$ is solvable in a neighborhood of $x_0$ and $t_0$, then the function
\begin{equation}\label{uM}
u(x,t,\e):=-2i\e\partial_x M_{1, 11}(x,t,\e),
\end{equation}
is, locally near $x_0$ and $t_0$, a solution to the KdV equation. Modifying the function $r(\lambda;\e)$ in the jump matrix leads to different KdV solutions.
If $r(\lambda;\e)$ is the reflection coefficient corresponding to the Schr\"odinger equation \[\e^2\frac{d^2}{dx^2}f=u_0(x)f,\] (with $u_0$ satisfying Assumptions \ref{assumptions}, but also under much milder conditions)
then $u$ is the solution to the Cauchy problem (\ref{KdV}) for KdV with initial data $u_0$. In this case the RH problem is solvable for all $x\in\mathbb R$ and $t\geq 0$; its solution can be constructed using fundamental solutions to the Schr\"odinger equation.

Our main task in what follows, is to approximate the RH solution $M$ asymptotically for $\e\to 0$ and $x$ close to the trailing edge. To this end, we need estimates for the reflection coefficient $r(\lambda;\e)$ as $\e\to 0$.

\subsection{Asymptotics for the reflection coefficient}\label{section refl}
Semiclassical asymptotics for the reflection coefficient as $\e\to 0$ were obtained in \cite{Ramond, FujiieRamond} for initial data $u_0$ that are such that Assumptions \ref{assumptions} hold: we have
\begin{align}
&r(\lambda;\e)= ie^{-\frac{2i}{\e}\rho(\lambda)}(1+\bigO(\e)),&\mbox{ for $\lambda\in(-1,0)$,}\\
&r(\lambda;e)=\bigO(e^{-\frac{c}{\e}}),&\mbox{ for $\lambda\in(-\infty,-1)$,}
\end{align}
where
$\rho$ is given by
\begin{equation}\label{def rho}
\rho(\lb)=f_L(\lambda)\sqrt{-\lambda}+\int_{-\infty}^{f_L(\lambda)}[\sqrt{u_0(x)-\lambda}-\sqrt{-\lambda}]dx.
\end{equation}
Moreover, there is a sector containing $\mathbb R^-$ such that $r(\lambda;\e)$ is analytic for $\lambda$ in this sector \cite{FujiieRamond}.
Let us write
\begin{equation}\label{def kappa}
\kappa(\lambda;\e)=-ir(\lb;\e)e^{\frac{2i}{\e}\rho(\lambda)}, \qquad\mbox{ for $\lb\in\Omega_+$},
\end{equation}
where
\begin{equation}\label{def Omega}\Omega_+:=\{\lb\in\mathbb C:-1-\delta<\Re\lb<-\delta, 0<\Im\lambda<\delta, |\lambda+1|>\frac{\delta}{2}\}.\end{equation}
For sufficiently small $\delta>0$, we have the following as $\e\to 0$, see \cite{Ramond} and \cite[Section 2.2]{CG2}
\begin{align}
&\label{kappa1}\kappa(\lb;\e)=1+\bigO(\e),&\mbox{ uniformly for $\lb\in\Omega_+$,}\\
&\label{ref2}r(\lb;\e)=\bigO(e^{-\frac{c}{\e}}), &\mbox{ for $\lb<-1-\delta$, $c>0$.}
\end{align}
In addition
\begin{equation}
\label{tau}
1-|r(\lambda;\epsilon)|^2\sim e^{-\frac{2}{\epsilon}\tau(\lambda)}, \qquad\mbox{ for $-1+\delta<\lambda<0$,}
\end{equation}
with
\begin{equation}\label{deftau}
\tau(\lambda)=\int_{f_L(\lambda)}^{f_R(\lambda)}\sqrt{\lambda - u_0(x)}dx,
\end{equation}
and $f_R$ is the inverse of the increasing part of the initial data $u_0$.
Those estimates for the reflection coefficient are essential for the RH analysis we will perform in the next section.

\section{Asymptotic analysis of the RH problem}\label{section 3}

We study the RH problem for $M$ asymptotically in the small dispersion limit when $x$ is close to the trailing edge.
Our approach is based on the Deift/Zhou steepest descent method \cite{DZ1} which has been applied to the KdV RH problem
in \cite{DVZ}. We follow roughly the same lines as in \cite{CG, CG2}, where the RH problem was studied near the point
of gradient catastrophe and near the leading edge. Compared to the situation near the leading edge,
the main difference here is the construction of the local parametrix near the point $v$,
which will be built out of Hermite polynomials.
\subsection{$G$-function and transformation $M\mapsto T$}
Let us define, for $t_c<t<T$, $G=G(\lambda;x,t)$ by
\begin{equation}
G(\lambda)=\frac{(u - \lambda)^{1/2}}{\pi}
\int_{u}^0\frac{\rho(\eta)-\alpha(\eta;x,t)}{\sqrt{\eta-u}(\eta
-\lambda)}d\eta,\label{def G}
\end{equation}
where $u=u(t)$ is the solution to (\ref{trailing1})-(\ref{trailing3}), and $\alpha$ and $\rho$ given by (\ref{def rho}) and (\ref{def alpha1}).
The square root $(u - \lambda)^{1/2}$ is analytic for $\lambda\in\mathbb C\setminus [u,+\infty)$ and positive for $\lambda<u$.
$G$ has the asymptotic behavior
\begin{equation}\label{g0}G(\lambda;x)=G_1(x,t)\frac{1}{(-\lambda)^{1/2}}+\bigO(\lambda^{-1}),\qquad \mbox{ as $\lambda\to\infty$},\end{equation}
with
\begin{eqnarray}\label{g1}
{G}_1(x,t)=\frac{1}{\pi}
\int_{u}^0\dfrac{\rho(\eta)-\alpha(\eta;x,t)}{\sqrt{\eta-u}}d\eta.
\end{eqnarray}
Since $\rho(\eta)$ is independent of $x$, we have
\begin{equation}
\partial_x G_1(x,t)=-\frac{1}{\pi}\int_{u}^0\sqrt{\frac{-\eta}{\eta-u}}d\eta=\dfrac{u}{2}.
\label{dg1}
\end{equation}
$G$ is analytic for $\lambda\in\mathbb C\setminus [u,+\infty)$ and it satisfies the properties
\begin{align}
&\label{Gjump1}G_+(\lb)+G_-(\lb)-2\rho(\lb)+2\alpha(\lb)=0, &\mbox{ as $\lb\in (u,0)$},\\
&\label{Gjump2}G_+(\lb)+G_-(\lb)=0, &\mbox{ as $\lb\in (0,+\infty)$}.
\end{align}
In order to modify the jumps for $M$ in a convenient way without losing its normalization at infinity, we define
\begin{equation}
\label{def T}
T(\lambda)=M(\lambda)e^{-\frac{i}{\e}G(\lambda;x)\sigma_3}.
\end{equation}
Then, using the RH conditions for $M$, we obtain the following RH problem for $T$.

\subsubsection*{RH problem for $T$}
\begin{itemize}
\item[(a)] $T$ is analytic in $\mathbb C\setminus \mathbb R$,
\item[(b)] $T_+(\lambda)=T_-(\lambda)v_T(\lambda)$, \ \ \ as
$\lambda\in\mathbb R$, with
\begin{equation}\label{vT}
v_T(\lambda)=\begin{cases}
\begin{array}{ll}
\sigma_1, &\mbox{ as $\lambda>0$,}\\[0.8ex]
\begin{pmatrix}1&r(\lb)e^{\frac{2i}{\e}(
G(\lb)+\alpha(\lb))}\\
-\bar r(\lb)e^{-\frac{2i}{\e}(G(\lb)+\alpha(\lb))} &1-|r(\lb)|^2
\end{pmatrix},&\mbox{ as $\lambda\in (-\infty,u)$,}\\[3ex]
\begin{pmatrix}e^{-\frac{i}{\e}(G_+(\lb)-G_-(\lb))}
&r(\lb)e^{\frac{2i}{\e}\rho(\lb)}\\
-\bar{r}(\lb)e^{-\frac{2i}{\e}\rho(\lb)}&(1-|r(\lb)|^2)e^{\frac{i}{\e}(
G_+(\lb)-G_-(\lb))}
\end{pmatrix},&\mbox{ as $\lambda\in (u,0)$,}
\end{array}
\end{cases}
\end{equation}
\item[(c)]
$T(\lambda)=(I+\bigO(\lambda^{-1}))\begin{pmatrix}1&1\\
i\sqrt{-\lambda}&-i\sqrt{-\lambda}\end{pmatrix}$ as
$\lambda\to\infty$.
\end{itemize}
By (\ref{uM}), (\ref{dg1}), and (\ref{def T}), we obtain the identity
\begin{equation}\label{uS1}
u(x,t,\e)=u-2i\e \partial_{x}T_{1,11}(x,t,\e),
\end{equation}
where
\begin{equation}
T_{11}(\lb;x,t,\e)=1+\frac{T_{1,11}(x,t,\e)}{\sqrt{-\lb}}+\bigO(\lb^{-1}), \qquad \mbox{ as $\lb\to\infty$.}
\end{equation}

\subsection{Opening of lenses $T\mapsto S$}

\begin{figure}[t]
\begin{center}
    \setlength{\unitlength}{1.2mm}
    \begin{picture}(137.5,26)(22,11.5)
        \put(112,25){\thicklines\circle*{.8}}
        \put(45,25){\thicklines\circle*{.8}}
        \put(47,26){$-1-\delta$}
        \put(112,26){$0$}
        \put(90,25){\thicklines\circle*{.8}} \put(74,26){$u$}\put(90,26){$v$}
        \put(102,25){\thicklines\vector(1,0){.0001}}
        \put(90,25){\line(1,0){45}}
        \put(124,25){\thicklines\vector(1,0){.0001}}
        \put(22,25){\line(1,0){23}}
        \put(35,25){\thicklines\vector(1,0){.0001}}
        \put(75,25){\line(1,0){15}}
        \put(84,25){\thicklines\vector(1,0){.0001}}
        \qbezier(45,25)(60,45)(75,25) \put(61,35){\thicklines\vector(1,0){.0001}}
        \qbezier(45,25)(60,5)(75,25) \put(61,15){\thicklines\vector(1,0){.0001}}
\end{picture}
\caption{The jump contour $\Sigma_S$ after opening of the lens}
\end{center}
\label{figure: lens}
\end{figure}
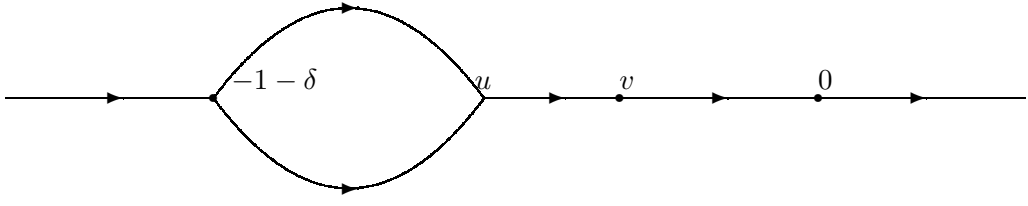

Let us define the function $\phi$ by
\[
\phi(\lambda;x,t)={G}(\lambda;x,t)-\rho(\lambda)+\alpha(\lambda;x,t),
\]
such that $\phi$ is analytic in a neighborhood of $(-1,u)$, with an analytic extension to the whole region $\Omega_+$ defined in (\ref{def Omega}).
Under the condition that (\ref{trailing1}) holds, it was shown in \cite[Lemma 3.2]{CG} that
\begin{equation}
\label{phi}
\phi(\lambda;x,t)=\sqrt{u-\lambda}(x-x^+)+\int_{\lambda}^{u}(f_L'(\xi)+6t)\sqrt{\xi-\lambda}d\xi,
\end{equation}
with, as usual, principal branches of the square roots.
It then follows that
\begin{equation}\label{phiprime}
\phi'(\lambda;x,t)=-\dfrac{1}{2\sqrt{u-\lambda}}(x-x^+)-\sqrt{u-\lambda}[6t+\theta(\lambda;u)],
\end{equation}
where $\theta(\lambda;u)$ is given by (\ref{theta}).
In a complex neighborhood of $v$ it is convenient to define $\widehat\phi$ as the analytic continuation of $-i\phi$ from the upper half plane,
\begin{multline}\label{hatphi}
\widehat\phi(\lambda;x,t)=\mp i\phi(\lambda; x,t)\\=-\sqrt{\lambda-u}(x-x^+)+\int_u^{\lambda}(f_L'(\xi)+6t)\sqrt{\lambda-\xi}d\xi, \qquad \mbox{ as $\pm\Im\lambda>0$.}
\end{multline}
We then have using (\ref{trailing2})-(\ref{trailing3}),
\begin{equation}
\widehat \phi(v;x^+,t)=0, \qquad \widehat\phi'(v;x^+,t)=0, \qquad \widehat\phi''(v;x^+,t)=\sqrt{v-u}\, \theta'(v;u)<0.
\end{equation}
Now we can express the jump matrix for $T$ in terms of $\phi$ and $\widehat\phi$:
\begin{align}\label{vT2}
&v_T(\lambda)=
\begin{pmatrix}e^{\frac{2}{\e}\widehat\phi(\lb)}
&i\kappa(\lambda;\e)\\
i\kappa^*(\lambda;\e)&(1-|r(\lb)|^2)e^{-\frac{2}{\e}\widehat\phi(\lb)}
\end{pmatrix},&\mbox{ as $\lambda\in (u,0)$,}
\\
&\label{vT3}
v_T(\lambda)=
\begin{pmatrix}1&i\kappa(\lb;\e)e^{\frac{2i}{\e}\phi(\lb)}\\
i\kappa^*(\lb;\e)e^{-\frac{2i}{\e}\phi^*(\lb)} &1-|r(\lb)|^2
\end{pmatrix},&\mbox{ as $\lambda\in [-1-\delta,u)$,}
\end{align}
where $\phi^*(\lambda)=\overline{\phi(\bar\lambda)}$, and $\phi(\lb)$ and $\phi^*(\lb)$ should be understood as the $+$ boundary values on $(-1-\delta,-1)$.

Similarly as in \cite[Proposition 3.1]{CG2}, $\phi$ and $\widehat\phi$ satisfy a number of inequalities.
\begin{proposition}\label{prop phi}
There exists a time $T>t_c$ such that for $t_c<t<T$, we have
\begin{equation}
\label{phiinequalities}
\begin{array}{ll}
\widehat\phi(\lambda;x^+)<0,&\mbox{ as $\lambda\in (u,0]\setminus\{v\}, $}\\[1.5ex]
\phi'(\lb;x^+,t)>0,&\mbox{ as $\lb\in[-1,u)$,}\\[1.5ex]
-\tau(\lambda)-\widehat\phi(\lb;x^+,t)<0,&\mbox{ as $\lb\in (u,0]$.}
\end{array}
\end{equation}
\end{proposition}
\begin{proof}
If (\ref{genericity}) holds, the function $6t_c+\theta(\lambda;u(t_c))$ has a double zero at $\lambda=u(t_c)$ and is strictly negative elsewhere on $[-1,0]$, see \cite{CG}. For $t$ slightly bigger than $t_c$, because of the smoothness of $u(t)$, $6t+\theta(\lambda;u(t))$ can have at most two zeroes on $[-1,0]$. It follows from (\ref{trailing2}) that one zero lies at $\lambda=v(t)$. Because of (\ref{trailing3}), another zero must lie in between $u$ and $v$. Consequently we have that $6t+\theta(\lambda;u(t))$ is negative for $\lb\in[-1,u)\cup (v,0]$ for sufficiently small times after the time of gradient catastrophe. Now it is straightforward to verify the first and the second inequaltity.

\medskip

For the last inequality, it is straightforward to verify by (\ref{deftau}) that
\[-\tau(\lambda)-\widehat\phi(\lb;x^-)=
-4t(\lambda-u)^{\frac{3}{2}}
+\int_{-1}^{u}\sqrt{\lambda-\xi}f'_L(\xi)d\xi
-\int_{-1}^{\lambda}\sqrt{\lambda-\xi}f'_R(\xi)d\xi,
\]
where  $f_R$ is the inverse function of the increasing part of the initial data $u_0(x)$.
The inequality follows directly because each of the three terms on the right hand side is negative.
\end{proof}

On the interval $(-1-\delta,u)$, we can factorize $v_T$:
\begin{equation}
v_T(\lb)=\begin{pmatrix}1&0\\
i{\kappa}^*(\lb)e^{-\frac{2i}{\e}
\phi_-(\lb)}&1
\end{pmatrix}\begin{pmatrix}1&i\kappa(\lb)e^{\frac{2i}{\e}\phi_+(\lb)}\\
0&1
\end{pmatrix}.
\end{equation}
Using this factorization, we can open lenses as shown in Figure \ref{figure: lens}. We choose the lenses such that the upper lens $\Sigma_1$ lies in $\Omega_+$ (for some fixed but sufficiently small $\delta>0$) and such that the lower lens $\Sigma_2$ is the complex conjugate of $\Sigma_1$.
Define
\begin{equation*}
S(\lambda)=\begin{cases}\begin{array}{ll}
T(\lambda)\begin{pmatrix}1&-i\kappa(\lb;\e)e^{\frac{2i}{\e}\phi(\lb;x)}\\
0&1
\end{pmatrix},&\mbox{ in the upper part of the lens,}\\[3ex]
T(\lambda)\begin{pmatrix}1&0\\
i\kappa^*(\lb;\e)e^{-\frac{2i}{\e}\phi^*(\lb;x)}&1
\end{pmatrix},&\mbox{ in the lower part of the lens},\\[3ex]
 T(\lambda), &\mbox{ elsewhere}.
\end{array}
\end{cases}
\end{equation*}
By the RH problem for $T$, we obtain modified RH conditions for $S$.
\subsubsection*{RH problem for $S$}
\begin{itemize}
\item[(a)] $S$ is analytic in $\mathbb C\setminus \Sigma_S$,
\item[(b)] $S_+(\lambda)=S_-(\lambda)v_S$ for $\lambda\in\Sigma_S$,
with
\begin{equation}\label{vS}
v_S(\lambda)=\begin{cases}
\begin{array}{lr}
\begin{pmatrix}1& i\kappa(\lb)e^{\frac{2i}{\e}\phi(\lb)}\\
0&1
\end{pmatrix},&\mbox{ on $\Sigma_1$},\\[3ex]
 \begin{pmatrix}1&0\\
i\kappa^*(\lb;\e)e^{-\frac{2i}{\e}\phi^*(\lb;x)}&1
\end{pmatrix},&\mbox{ on $\Sigma_2$,}\\[3ex]
\begin{pmatrix}e^{\frac{2}{\e}\widehat\phi(\lb;x)}
&i\kappa(\lambda;\e)\\
i\kappa^*(\lambda;\e)&(1-|r(\lb)|^2)e^{-\frac{2}{\e}\widehat\phi(\lb;x)}
\end{pmatrix},&\mbox{ as $\lambda\in (u,0)$,}\\[3ex]
\sigma_1,&\mbox{\hspace{-3cm} as
$\lambda\in(0,+\infty)$,}\\[3ex]
v_T(\lambda;x),&\mbox{\hspace{-3cm} as
$\lambda\in(-\infty,-1-\delta)$.}
\end{array}
\end{cases}
\end{equation}
\item[(c)] $S(\lambda)=(I+\bigO(\lambda^{-1}))\begin{pmatrix}1&1\\
i\sqrt{-\lambda}&-i\sqrt{-\lambda}\end{pmatrix}$ as
$\lambda\to\infty$.
\end{itemize}
Since $T(\lb)=S(\lb)$ for large $\lb$, we have, using (\ref{uS1}),
\begin{equation}\label{uS}
u(x,t,\e)=u-2i\e \partial_{x}S_{1,11}(x,t,\e),
\end{equation}
where
\begin{equation}
S_{11}(\lb;x,t,\e)=1+\frac{S_{1,11}(x,t,\e)}{\sqrt{-\lb}}+\bigO(\lb^{-1}), \qquad \mbox{ as $\lb\to\infty$.}
\end{equation}

The construction of $S$ has been organized in such a way that the jump matrices for $S$ decay uniformly to constant matrices, except in arbitrary small neighborhoods of $u$ and $v$.
\begin{proposition}\label{prop jumps}
There exists $T>t_c$ such that the following holds for $t_c<t<T$. For any fixed neighborhoods $U_u$ of $u$ and $U_v$ of $v$, there exists $\delta_0>0$ such that for $|x-x^+|<\delta_0$, the jump matrices $v_S(\lambda)=(I+\bigO(\e))v^{(\infty)}(\lambda;x,\e)$ uniformly on $\Sigma_S\setminus (U_u\cup U_v\cup\{0\})$ if $\e\to 0$, with
\begin{equation}\label{vinfty}
v^{(\infty)}(\lambda)=\begin{cases}
\sigma_1,&\mbox{ on $(0, +\infty)$,}\\
i\sigma_1,&\mbox{ on $(u,0)$,}\\
I, &\mbox{elsewhere}.\end{cases}
\end{equation}
\end{proposition}
\begin{proof}
See \cite[Proposition 3.3]{CG2}.
\end{proof}
For $x>x^+ +\frac{\e\ln\e}{2\sqrt{v-u}}$, it is straightforward to verify that the convergence is also uniform on $\mathbb R\cap U_v$ for sufficiently small $U_v$.

\subsection{Outside parametrix}

Let us first ignore small neighborhoods of $u$ and $v$, and uniformly small jumps. Then we need to solve the following RH problem for the outside parametrix.
\subsubsection*{RH problem for $P^{(\infty)}$}
\begin{itemize}
\item[(a)] $P^{(\infty)}:\mathbb C\setminus [u, +\infty) \to
\mathbb C^{2\times 2}$ is analytic, \item[(b)] $P^{(\infty)}$
satisfies the jump conditions
\begin{align}
&P_+^{(\infty)}=P_-^{(\infty)}\sigma_1, &\mbox{ on $(0, +\infty)$},\\
&P_+^{(\infty)}= iP_-^{(\infty)}\sigma_1, &\mbox{ on $(u,0)\setminus\{v\}$},\label{RHP Pinfty b}
\end{align}
\item[(c)]$P^{(\infty)}$ has the following behavior as $\lambda\to\infty$,
\begin{equation}P^{(\infty)}(\lambda)=(I+\bigO(\lambda^{-1}))
\begin{pmatrix}1&1\\ i(-\lambda)^{1/2} & -i(-\lambda)^{1/2}\end{pmatrix}
. \label{RHP Pinfty c}\end{equation}
\end{itemize}

This RH problem is solved by
\begin{equation}
\label{def Pinfty0}
P_0^{(\infty)}(\lambda)=(-\lambda)^{1/4}(u-\lambda)^{-\sigma_3
/4}N, \qquad N=\begin{pmatrix}1&1\\ i& -i
\end{pmatrix}.
\end{equation}
This solution is bounded near $v$. Depending on the value of \begin{equation}y:=2\sqrt{v-u}\, \frac{x-x^+}{\e\ln\e},\end{equation} we will need an
outside parametrix which has singular behavior near $v$.
Therefore we observe that
\begin{equation}
\label{def Pinftyk}
P_k^{(\infty)}(\lambda)=(-\lambda)^{1/4}(u-\lambda)^{-\sigma_3
/4}ND(\lambda)^{\mp k\sigma_3}, \qquad\mbox{ for $\pm\Im\lambda>0$,}
\end{equation}
is a solution to the RH problem for $P^{(\infty)}$ for any $k\in\mathbb N$,
with
\begin{equation}
\label{def D}
D(\lambda)=\dfrac{\sqrt{\lambda-u}-\sqrt{v-u}}{\sqrt{\lambda-u}+\sqrt{v-u}}.
\end{equation}
Indeed this follows from the fact that $D$ is analytic on $\mathbb C\setminus(-\infty,u]$, with
\begin{align}
&D_+(\lambda)D_-(\lambda)=1,&\mbox{ for $\lambda<u$,}\\
&D(\infty)=1.
\end{align}
Near $v$, we have that
$P_k^{(\infty)}(\lambda)\lambda^{\pm k\sigma_3}$ is bounded for $\pm\Im\lambda>0$.

\medskip

The appropriate choice for $k$ turns out to be as follows:
if $y\leq 0$, we take $k=0$, and if $y\geq 0$, we let $k$ be the non-negative integer closest to $y$ (for the half integers, we may choose $k=y-\frac{1}{2}$ or $k=y+\frac{1}{2}$), so $y-\frac{1}{2}\leq k\leq y+\frac{1}{2}$.
As $\lambda\to\infty$, we have
\begin{multline}\label{RHP Pinfty c2}P_k^{(\infty)}(\lambda)=\left(I+\frac{u}{4\lambda}\sigma_3
+\bigO(\lambda^{-2})\right)
\begin{pmatrix}1&1\\ i(-\lambda)^{1/2} & -i(-\lambda)^{1/2}\end{pmatrix}\\
\times \ \left(I-\frac{2ik\sqrt{v-u}}{(-\lambda)^{1/2}}\sigma_3+\bigO(\lambda^{-1})\right)
.\end{multline}

\subsection{Local parametrix near $u$}
In \cite[Section 3.5]{CG2}, the Airy function was used to construct a local parametrix in a neighborhood $U_u$ of the point $u$. If conditions (\ref{trailing1})-(\ref{trailing2}) hold, it was proved that the constructed parametrix satisfies the following properties.

\subsubsection*{RH problem for $P_u$}
\begin{itemize}
\item[(a)] $P_u$ is analytic in $\overline{U_u}\setminus\Sigma_S$.
\item[(b)] $P_u$ satisfies the jump conditions
\begin{align}
&P_{u,+}(\lb)=P_{u,-}(\lb)\begin{pmatrix}1&ie^{\frac{2i}{\e}\phi(\lb;x)}\\0&1\end{pmatrix},&\mbox{ as $\lb\in\Sigma_1\cap U_u$},\\
&P_{u,+}(\lb)=P_{u,-}(\lb)\begin{pmatrix}1&0\\ie^{\frac{-2i}{\e}\phi^*(\lb;x)}&1\end{pmatrix}, &\mbox{ as $\lb\in\Sigma_2\cap U_u$},\\
&P_{u,+}(\lb)=P_{u,-}(\lb)\begin{pmatrix}e^{\frac{-2i}{\e}\phi_+(\lb;x)}&i\\i&0\end{pmatrix}, &\mbox{ as $\lb>u$.}
\end{align}
\item[(c)] As $\e\to 0$ and $x\to x^+$ in such a way that $x-x^+=\bigO(\e^{2/3})$, $P_u$ matches with $P^{(\infty)}$ in the following way, \begin{equation}\label{matching P u}
P_u(\lambda)=P^{(\infty)}(\lambda)\left[I-\frac{is(\lambda;x)^2}{4\e (-\zeta(\lambda))^{1/2}}\sigma_3-\frac{s(\lambda;x)}{4\zeta}\sigma_1+\bigO(\e)\right], \qquad\mbox{ for $\lambda\in\partial U_u$.}
\end{equation}
Here $s(\lambda;x)$ is an analytic function of $\lambda\in \overline{U_u}$ with
\begin{equation}
\label{su}
s(u;x)=\left(-f'_L(u)-6t\right)^{-1/3}(x-x^+).
\end{equation}
\end{itemize}
In the (stronger) double scaling limit where $\e\to 0$ and $x\to x^+$ in such a way that $x-x^+=\bigO(\e\ln\e)$, it follows that
\begin{equation}\label{matching P u2}
P_u(\lambda)=P^{(\infty)}(\lambda)\left[I+\bigO(\e\ln^2\e)\right], \qquad\mbox{ for $\lambda\in\partial U_u$.}
\end{equation}
For the explicit construction of $P_u$ in terms of the Airy function, we refer to \cite[Section 3.5]{CG2}. It will be important in what follows that, in the limit $\e\to 0$, the jump matrices for $P_u$ are equal to those for $S$ up to an error of order $\bigO(\e)$. This follows from (\ref{vS}) and the asymptotic formulas (\ref{kappa1}) and (\ref{tau}) for the reflection coefficient.

\subsection{Local parametrix near $v$}
The construction of the local parametrix will depend on the value of
\begin{equation}
y=2\sqrt{v-u}\, \frac{x-x^+}{\e\ln\e}.
\end{equation}
If $y\leq -1$, there is no need to construct a separate local parametrix near $v$ since the jump matrix for $S$ is equal to the one for $P^{(\infty)}$ up to an $\bigO(\e\ln^2\e)$ error; in this case we write for notational convenience $P_v(\lambda)=I$. Let us assume now that $y>-1$.
The aim of this section is to construct a local parametrix in a small neighborhood $U_v$ of $v$ having approximately the same jump property as $S$ has near $v$ (in the limit $\e\to 0$), and matching with the outside parametrix at $\partial U_v$. Substituting the semiclassical asymptotics (\ref{kappa1}) and (\ref{tau}) for the reflection coefficient, the jump matrix for $S$ given in (\ref{vS}) behaves as follows when $\e\to 0$,
\begin{equation}
v_S(\lambda)=(I+\bigO(\e))\begin{pmatrix}e^{\frac{2}{\e}\widehat\phi(\lb;x)}
&i\\
i&0
\end{pmatrix}.\end{equation} This brings us to the RH problem for the local parametrix.
\subsubsection*{RH problem for $P_v$}
\begin{itemize}
\item[(a)] $P$ is analytic in $\overline U_v\setminus \mathbb R$,
\item[(b)] $P$ has the jump condition\begin{equation}\label{RHP P:b}
P_+(\lambda)=P_-(\lambda)\begin{pmatrix}e^{\frac{2}{\e}\widehat\phi(\lb;x)}
&i\\
i&0
\end{pmatrix},\end{equation}
\item[(c)] if we let $\e\to 0$ and simultaneously $x\to x^+$ in such a way that
\begin{equation}
|x-x^+|<M\e |\ln\e|,
\end{equation}
we have
\begin{equation}\label{RHP P:c}
P(\lambda)=P^{(\infty)}(\lambda)(1+o(1)), \qquad \mbox{ as $\lambda\in\partial U_v$.}
\end{equation}
\end{itemize}

We will construct $P$ explicitly in terms of a model RH problem built out of Hermite polynomials.

\subsubsection{Model RH problem}
Define $\Psi$ by
\begin{multline}\label{def Psi}
\Psi(\zeta;k)={\small \begin{pmatrix}
       \frac{1}{h_k}H_k(\zeta) & \frac{1}{2\pi i h_k}            \int_{\mathbb R} \frac{H_k(u) e^{-u^2}}{u-\zeta}\,du
        \\[3ex]
         -2\pi i h_{k-1} H_{k-1}(\zeta) &-h_{k-1}
            \int_{\mathbb R}\frac{H_{k-1}(u) e^{-u^2}}{u-\zeta}\,du
    \end{pmatrix}e^{-\frac{\zeta^2}{2}\sigma_3}},
    \\ \mbox{for $\zeta\in\mathbb C \setminus\mathbb
    R$,}
\end{multline}
where $H_k$ denotes the degree $k$ Hermite polynomial, orthonormal
with respect to the weight $e^{-x^2}$ on $\mathbb R$. The leading coefficient of the normalized polynomial $H_k$ is equal to
\begin{equation}
\label{hk}
h_k=\frac{2^{k/2}}{\pi^{1/4}\sqrt{k!}},
\end{equation}
 and we agree $H_{-1}=h_{-1}=0$.
$\Psi=\Psi(\zeta;k)$ solves the following RH problem for $k\in\mathbb N\cup\{0\}$, which is a particular case of the RH problem for orthogonal polynomials discovered in \cite{FIK}.

\subsubsection*{RH problem for $\Psi$}
\begin{itemize}
\item[(a)]$\Psi:\mathbb C\setminus\mathbb R\to\mathbb C^{2\times
2}$ is analytic. \item[(b)] $\Psi$ has continuous boundary values for $\zeta\in\mathbb R$, related by the condition
\begin{equation}\label{jump
Psi}\Psi_+(\zeta)=\Psi_-(\zeta)\begin{pmatrix}1&1\\0&1
\end{pmatrix}.\end{equation}
\item[(c)]As $\zeta\to\infty$, we have
\begin{equation}\label{RHP Psi:c}
\Psi(\zeta)={\small\left(I-\dfrac{1}{\zeta}
\begin{pmatrix}
0&\frac{1}{2\pi i}h_{k}^{-2}\\
2\pi i h_{k-1}^2&0\end{pmatrix}+\frac{1}{\zeta^2}\begin{pmatrix}\frac{k(k-1)}{4}&0\\0&-\frac{k(k+1)}{4}\end{pmatrix}+
\bigO\left(\zeta^{-3}\right)\right)}\zeta^{k\sigma_3}e^{-\frac{\zeta^2}{2}\sigma_3}.\end{equation}
\end{itemize}
The further terms in the large-$\zeta$ expansion can be calculated explicitly, but are unimportant for us.

\subsubsection{Modified model RH problem}

In order to have a model RH problem which resembles the RH problem for $P$, we let
\begin{equation}
\Phi(\zeta;k, s)=
\begin{cases}
e^{\frac{s^2}{2}\sigma_3} e^{-\frac{i\pi}{4}\sigma_3}\Psi(\zeta+s;k)e^{\frac{i\pi}{4}\sigma_3}i\sigma_1,&\mbox{ as $\Im\zeta>0$,}\\
e^{\frac{s^2}{2}\sigma_3}e^{-\frac{i\pi}{4}\sigma_3}\Psi(\zeta+s;k)e^{\frac{i\pi}{4}\sigma_3},&\mbox{ as $\Im\zeta<0$.}
\end{cases}
\end{equation}
It is straightforward to check the RH conditions that are satisfied by $\Phi$.

\subsubsection*{RH problem for $\Phi$}
\begin{itemize}
\item[(a)]$\Phi:\mathbb C\setminus\mathbb R\to\mathbb C^{2\times
2}$ is analytic. \item[(b)] For $x\in\mathbb R$,
\begin{equation}\label{jump
Phi}\Phi_+(x)=\Phi_-(x)\begin{pmatrix}1&i\\i&0
\end{pmatrix}.\end{equation}
\item[(c)]$\Phi$ behaves as follows as $\zeta\to\infty$,
\begin{align}\label{RHP Phi:c1}&\Phi(\zeta)=\widehat\Phi(\zeta)\zeta^{k\sigma_3}
e^{-\frac{\zeta^2+2s\zeta}{2}\sigma_3}i\sigma_1, &\mbox{ as
$\zeta\to\infty$, $\Im\lambda>0$},\\
\label{RHP Phi:c2}&\Phi(\zeta)=\widehat\Phi(\zeta)\zeta^{k\sigma_3}
e^{-\frac{\zeta^2+2s\zeta}{2}\sigma_3}, &\mbox{ as
$\zeta\to\infty$, $\Im\lambda<0$,}\end{align}
where $\widehat\Phi$ has the asymptotic expansion
\begin{equation}\label{expansion Phi}
\widehat\Phi(\zeta)=I+\frac{1}{\zeta}\widehat\Phi_1+\frac{1}{\zeta^2}\widehat\Phi_2+\bigO(\zeta^{-3}),
\end{equation}
with
\begin{equation}\label{Phihat1}
\widehat\Phi_1=
\begin{pmatrix}ks&e^{s^2}\frac{1}{2\pi  h_k^2}\\
e^{-s^2}2\pi h_{k-1}^2&-ks\end{pmatrix},
\end{equation}
\begin{equation}\label{Phihat2}
\widehat\Phi_2=\begin{pmatrix}\frac{k(k-1)}{2}(s^2+\frac{1}{2})&-(k+1)se^{s^2}\frac{1}{2\pi  h_k^2}\\
(k-1)se^{-s^2}2\pi h_{k-1}^2&\frac{k(k+1)}{2}(s^2-\frac12))\end{pmatrix}.
\end{equation}
\end{itemize}

\subsubsection{Construction of the parametrix}

We construct $P$ in the form
\begin{equation}\label{def P}
P(\lambda)=E_k(\lambda)\Phi(\e^{-1/2}\zeta(\lambda);k, \e^{-1/2}s(\lambda))e^{\pm \frac{1}{\e}\hat\phi(\lambda)\sigma_3},\qquad\mbox{ as $\pm\Im\lambda>0$,}
\end{equation}
where $\zeta$ is a real conformal mapping which maps $v$ to $0$, $s$ is analytic near $v$, and
\begin{equation}\label{def E}
E_k(\lambda)=\begin{cases}-iP_k^{(\infty)}(\lambda)\sigma_1\zeta(\lambda)^{-k\sigma_3}\e^{-\frac{\Delta_k}{2}\sigma_3},&\mbox{ as $\Im\lambda>0$,}\\
P_k^{(\infty)}(\lambda)\zeta(\lambda)^{-k\sigma_3}\e^{-\frac{\Delta_k}{2}\sigma_3},&\mbox{ as $\Im\lambda<0$,}
\end{cases}
\end{equation}
which defines $E_k$ analytically near $v$, with
\begin{equation}E_k(v)=(-v)^{1/4}(v-u)^{-\sigma_3
/4}e^{-\frac{\pi i}{4}\sigma_3}N\zeta'(v)^{-k\sigma_3}\e^{\frac{-\Delta_k}{2}\sigma_3}4^{-k\sigma_3}(v-u)^{-k\sigma_3}.
\end{equation}
As before we take
\begin{equation}
\label{def y2}
y=\frac{2}{\e\ln\e}\sqrt{v-u}(x-x^+),
\end{equation}
and let $k$ be the non-negative integer closest to $y$, and $\Delta_k:=y-k$.
Using (\ref{RHP P:b}), (\ref{jump Phi}), and (\ref{def P}), one easily verifies that $P$ satisfies the correct jump condition near $v$.

\medskip

In order to have the right matching (\ref{RHP P:c}) between $P$ and $P^{(\infty)}$, we need to exploit the remaining freedom in choosing the conformal mapping $\zeta(\lambda)$ and the analytic function $s(\lambda)$ in such a way that
\begin{equation}\label{condition zeta s}
2\widehat\phi(\lambda; x)=-\zeta(\lambda)^2-2s(\lambda; x)\zeta(\lambda)-y \e\ln\e.
\end{equation}
If we take a look at $\widehat\phi$ in (\ref{hatphi}), we see that
\begin{multline}\label{hatphi1}
\widehat\phi(\lambda;x)=-\sqrt{v-u}(x-x^+)-(\sqrt{\lambda-u}-\sqrt{v-u})(x-x^+)\\+\int_u^{\lambda}(f_L'(\xi)+6t)\sqrt{\lambda-\xi}d\xi.
\end{multline}
Let us now define $\zeta$ by
\begin{equation}\label{def f}
\zeta(\lambda)^2=-2\widehat\phi(\lambda; x^+)=-2\int_u^{\lambda}(f_L'(\xi)+6t)\sqrt{\lambda-\xi}d\xi,
\end{equation}
so that
\begin{equation}\label{fv}
\zeta(v)=0, \qquad \zeta'(v)=(v-u)^{1/4}(-\theta'(v;u))^{1/2}>0.
\end{equation}
Furthermore let
\begin{equation}\label{def s}
s(\lambda;x)=\frac{1}{\zeta(\lambda)}(\sqrt{\lambda-u}-\sqrt{v-u})(x-x^+),
\end{equation}
which defines $s$ analytically near $v$ with
\begin{equation}
s(v; x)=\frac{x-x^+}{2(v-u)^{3/4}(-\theta'(v;u))^{1/2}}.
\end{equation}
Now by (\ref{hatphi1}), (\ref{def f}), and (\ref{def s}), it follows that (\ref{condition zeta s}) holds.

\medskip

We constructed the parametrix in such a way that for $\lambda\in\partial U_v$,
\begin{multline}
P(\lambda)P^{(\infty)}(\lambda)^{-1}=E_k(\lambda)
\left(I+\frac{\e^{1/2}}{\zeta(\lambda)}\widehat\Phi_1+\frac{\e}{\zeta(\lambda)^2}\widehat\Phi_2+\bigO(\e^{3/2})\right)E_k(\lambda)^{-1}.
\end{multline}
If we choose $k$ to be the non-negative integer closest to $y$, then it follows from (\ref{def E}) that $P$ and $P^{(\infty)}$ have a good matching at $\partial U_v$ as long as $y$ is not to close to a half positive integer. In order to have a uniform matching, also when $y$ approaches a half integer, we need to improve the parametrices.

\subsection{Improvement of parametrices}

The goal of this section is to improve the local parametrices and the outside parametrix in such a way that they match uniformly for $y$ bounded. This section is inspired by \cite{BertolaLee}. Let us define, if $y>-1$,
\begin{align}
&\label{Pinftyimp}\widetilde P^{(\infty)}(\lambda)=\left(I+\frac{C(x,t,\e)}{\lambda-v}\right)P^{(\infty)}(\lambda),\\
&\widetilde P_u(\lambda)=\left(I+\frac{C(x,t,\e)}{\lambda-v}\right)P_u(\lambda),
\end{align}
where $C$ is a nilpotent matrix (so that the determinant of $Q$ is identically $1$) which we will determine below.
With those improved definitions of the parametrices, $\widetilde P^{(\infty)}$ and $\widetilde P_u$ satisfy the same RH conditions as before (see (\ref{RHP Pinfty b})-(\ref{RHP Pinfty c}) for $P^{(\infty)}$). If $C$ is bounded for small $\e$, we have by (\ref{matching P u2})
\begin{equation}
\widetilde P_u(\lambda)\widetilde P^{(\infty)}(u)^{-1}=I+\bigO(\e\ln^2\e)),\qquad\mbox{ as $\lambda\in\partial U_u$,}
\end{equation}
if $\e\to 0$.
Next we define the improved local parametrix near $v$ as follows,
\begin{multline}
\label{def P imp1}
\widetilde P_v(\lambda)=\widetilde E_k(\lambda)
\begin{pmatrix}1&-\delta\e^{\frac{1}{2}}e^{s^2}\frac{1}{2\pi h_k^2\zeta(\lambda)}\\
-(1-\delta)\e^{\frac{1}{2}}e^{-s^2}\frac{2\pi h_{k-1}^2}{\zeta(\lambda)}&1\end{pmatrix}\\
\times\quad
\Phi(\e^{-1/2}\zeta(\lambda);k, \e^{-1/2}s(\lambda))e^{\pm \frac{1}{\e}\hat\phi(\lambda)\sigma_3},
\end{multline}
as $\pm\Im\lambda>0$, with
\begin{equation}
\widetilde E_k(\lambda)=\left(I+\frac{C(x,t,\e)}{\lambda-v}\right)E_k(\lambda),
\end{equation}
and with
\begin{equation}\delta=1,\qquad\mbox{ if $k\leq y\leq k+\frac{1}{2}$}, \qquad \delta=0,\qquad \mbox{ if $k-\frac{1}{2}\leq y<k$}.\end{equation}
In any case one of the off-diagonal entries of the second factor on the right hand side of (\ref{def P imp1}) vanishes.
If we let $\e\to 0$ in such a way that $y$ remains bounded, we have that $s(\lambda;x)=\bigO(\e\ln\e)$, and it follows from
(\ref{RHP Phi:c1})-(\ref{RHP Phi:c2}), (\ref{def P}), and (\ref{def E}) that
\begin{equation}
\widetilde P_v(\lambda)\widetilde P^{(\infty)}(\lambda)^{-1}=I+\bigO(\e^{1/2+|\Delta_k|}),\qquad\mbox{ as $\lambda\in\partial U_v$.}
\end{equation}
Note that $\widetilde E_k$ is not analytic at $v$: it has a simple pole at $v$.
However we can choose the matrix $C$ such that $\widetilde P$ is bounded near $v$.
After a straightforward calculation, it turns out that this is the case if
\begin{equation}\label{def C1}
C=E_k(v)\begin{pmatrix}0&\delta d_1\\(1-\delta)c_1&0\end{pmatrix}
\left[E_k(v)\begin{pmatrix}1&-\delta d_2\\-(1-\delta)c_2&1\end{pmatrix}-E_k'(v)\begin{pmatrix}0&\delta d_1\\(1-\delta)c_1&0\end{pmatrix}\right]^{-1},
\end{equation}
with
\begin{align}
&c_1=c_1(k;\e)=\frac{\e^{\frac{1}{2}}e^{-s^2(v)}2\pi h_{k-1}^2}{\zeta'(v)},\\
&c_2=c_2(k;\e)=-\frac{c_1\zeta''(v)}{2\zeta'(v)}-2s(v)s'(v)c_1,\\
&d_1=d_1(k;\e)=\frac{\e^{\frac{1}{2}}e^{s^2(v)}}{2\pi h_k^2\zeta'(v)},\\
&d_2=d_2(k;\e)=-\frac{d_1\zeta''(v)}{2\zeta'(v)}+2s(v)s'(v)d_1.
\end{align}
Note that $c_1(k+1;\e)d_1(k;\e)=\frac{\e}{\zeta'(v)^2}$, and that the matrix $C$ bounded for small $\epsilon$.
In particular, writing $E=E_k$, we have
\begin{multline}
\label{def C imp1}
C_{12}=-(1-\delta)\frac{c_1 E_{12}(v)^2}{\det E(v)-c_1\left[E_{12}'(v)E_{22}(v)-E_{12}(v)E_{22}'(v)\right]}
\\
+\delta\frac{d_1 E_{11}(v)^2}{\det E(v)-d_1\left[E_{11}(v)E_{21}'(v)-E_{11}'(v)E_{21}(v)\right]}.
\end{multline}
Observe that
\begin{align}
&\det E(v)=-2i\sqrt{-v},\\
&E_{11}(v)^2=-i\sqrt{\frac{-v}{v-u}}\left(4(v-u)\right)^{-2k}\zeta'(v)^{-2k}\e^{-\Delta_k},\\
&E_{12}(v)^2=-i\sqrt{\frac{-v}{v-u}}\left(4(v-u)\right)^{2k}\zeta'(v)^{2k}\e^{\Delta_k},\\
&E_{11}(v)E_{21}'(v)-E_{11}'(v)E_{21}(v)=2i\sqrt{-v}\left(4\zeta'(v)(v-u)\right)^{-2k-1}\zeta'(v)\e^{-\Delta_k},\\
&E_{12}'(v)E_{22}(v)-E_{12}(v)E_{22}'(v)=2i\sqrt{-v}\left(4\zeta'(v)(v-u)\right)^{2k-1}\zeta'(v)\e^{\Delta_k}.
\end{align}
This leads to
\begin{multline}
\label{def C imp1b}
C_{12}(k,\Delta_k)=-2(1-\delta)\sqrt{v-u}\frac{c_1\gamma^{2k-1}\zeta'(v)\e^{\Delta_k}}{1+c_1\gamma^{2k-1}\zeta'(v)\e^{\Delta_k}}\\
+2\delta\sqrt{v-u}\frac{d_1\gamma^{-2k-1}\zeta'(v)\e^{-\Delta_k}}{1+d_1\gamma^{-2k-1}\zeta'(v)\e^{-\Delta_k}}+\bigO(\e^2\ln^2\e),
\end{multline}
in the double scaling limit where $\e\to 0$ and $x-x^+=\bigO(\e\ln\e)$,
with
\begin{equation}\label{gamma}
\gamma=4\zeta'(v)(v-u).
\end{equation}

\subsection{Final RH problem}

Now we define $R$ in such a way that it has jumps that are uniformly $I+\bigO(\e^{1/2})$ in the double scaling limit where $\e\to 0$, $x-x^+=\bigO(\e\ln\e)$, also when $y$ is close to a half integer. Therefore we let
 \begin{equation}
\label{def R imp}
R(\lambda; x,\e)=\begin{cases}
S(\lambda;x)\widetilde P^{(\infty)}(\lambda)^{-1}, &\mbox{ as $\lambda\in\mathbb C\setminus (U_u\cup U_v),$}\\
S(\lambda;x)\widetilde P_u(\lambda;x)^{-1}, &\mbox{ as $\lambda\in U_u,$}\\
S(\lambda;x)\widetilde P_v(\lambda;x)^{-1}, &\mbox{ as $\lambda\in U_v$,}
\end{cases}
\end{equation}
where we use the notations $\widetilde P_u=P_u$ and $\widetilde P_v=I$ in the case where $y\leq -1$.
Obviously, $R$ is analytic in $\mathbb C\setminus\left(\Sigma_S\cup\partial U_u\cup\partial U_v\right)$.
On $\partial U_u\cup\partial U_v$, the jump matrix for $R$ is close to the identity matrix because of the matching of the (improved) local parametrices with the outside parametrices. Outside the disks $U_u$ and $U_v$, the decay of the jump matrix to $I$ is inherited from the decay of the jump matrix for $S$. Inside the disks, there are residual jumps on $\Sigma_S$ because the jump matrices for the local parametrices are not exactly the same as the ones for $S$. However, using the small $\e$ asymptotics for the reflection coefficient, one verifies easily that they are of the form $I+\bigO(\e)$ as $\e\to 0$.

\subsubsection*{RH problem for $R$}
\begin{itemize}
\item[(a)] $R$ is analytic in $\mathbb C\setminus(\Sigma_S\cup \partial U_u\cup \partial U_v)$.
\item[(b)] $R$ has the jump condition $R_+(\lambda)=R_-(\lambda)v_R(\lambda)$ for $\lambda\in \Sigma_S\cup \partial U_u\cup \partial U_v$, where
 \begin{align}
 &v_R(\lambda)=I+\bigO(e^{-\frac{c}{\e}}), &\mbox{ for $\lambda\in\Sigma_S\setminus(U_u\cup U_v)$,}\\
 &v_R(\lambda)=I+\bigO(\e), &\mbox{ for $\lambda\in\Sigma_S\cap(U_u\cup U_v)$,}\\
  &v_R(\lambda)=I+\bigO(\e\ln^2\e), &\mbox{ for $\lambda\in\partial U_u$,}\\
   &v_R(\lambda)=I+\bigO(\e^{\frac{1}{2}+|\Delta_k|}), &\mbox{ for $\lambda\in\partial U_v$,}
 \end{align}
 in the double scaling limit where $\e\to 0$ in such a way that $y$ remains bounded. We choose the clockwise orientation for $\partial U_u$ and $\partial U_v$.
\item[(c)] As $\lambda\to\infty$, we have
\begin{equation}\label{RHP R:c}R(\lambda)=I+\frac{R_1}{\lambda}+\bigO(\lambda^{-2}).\end{equation}\end{itemize}

All jumps are thus of the form $I+\bigO(\e^{1/2})$, and if $y$ is a half integer or if $y\leq -1$, the situation improves to $I+\bigO(\e\ln^2\e)$.
Taking a closer look at the jump for $R$ on $\partial U_v$, and using $s(\lambda;x)=\bigO(\e\ln\e)$, we obtain using (\ref{def P imp1}) and (\ref{expansion Phi}) that
\begin{eqnarray*}
v_R(\lambda)&=&\widetilde P_v(\lambda)\widetilde P^{(\infty)}(\lambda)^{-1}\nonumber \\
&=&\widetilde E_k(\lambda)\left(I+\frac{\e^{1/2}}{\zeta(\lambda)}\begin{pmatrix}0&(1-\delta)\frac{1}{2\pi h_k^2}\\
\delta 2\pi h_{k-1}^2
&0\end{pmatrix} \right.\\
&&\qquad\qquad\qquad \left. +\frac{\e}{\zeta(\lambda)^2}\begin{pmatrix}\frac{k(k-2)}{2}&0\\
0&-\frac{k(k-2)}{2}\end{pmatrix}+\bigO(\e^{\frac{3}{2}}\ln\epsilon)\right)\widetilde E_k(\lambda)^{-1}.
\end{eqnarray*}
By (\ref{def C1}), this leads to
\begin{equation}\label{vR3}
v_R(\lambda)=I+\frac{\e^{1/2}}{\zeta(\lambda)}E_k(\lambda)\begin{pmatrix}0&(1-\delta) \frac{1}{2\pi h_k^2}\\
\delta 2\pi h_{k-1}^2
&0\end{pmatrix}E_k(\lambda)^{-1}+\bigO(\e).
\end{equation}
We can write
\[v_R(\lambda)=I+\e^{\frac{1}{2}+|\Delta_k|}v_R^{(1)}(\lambda)+\bigO(\e\ln^2\e),\qquad\mbox{ for $\lambda\in \Sigma_S\cup \partial U_u\cup \partial U_v$},\] uniformly for $\e\to 0$ and $y$ bounded.
By a standard procedure for small-norm RH problems \cite{DKMVZ1, KMM}, it follows that $R$ has a similar expansion in the double scaling limit:
\begin{equation}\label{Rexpansion}
R(\lambda)=I+\e^{1/2+|\Delta_k|}R^{(1)}(\lambda)+\bigO(\e\ln^2\e).
\end{equation}
Compatibility of (\ref{vR3}) with (\ref{Rexpansion}) and the jump condition $R_+=R_-v_R$ on $\partial U_v$ gives the relation
\begin{equation}
R_+^{(1)}(\lambda)-R_-^{(1)}(\lambda)=v_R^{(1)}(\lambda),\qquad\mbox{for $\lambda\in\partial U_v$.}
\end{equation}
In addition $R^{(1)}$ is analytic in $\mathbb C\setminus\partial U_v$, and $R^{(1)}(\infty)=0$. The unique function which satisfies those (additive) RH conditions is given by
\begin{align}
&\label{R1out}R^{(1)}(\lambda)=\frac{\Res(v_R^{(1)};v)}{\lambda-v},&\mbox{ if $\lambda\in\mathbb C\setminus U_v$,}\\
&R^{(1)}(\lambda)=\frac{\Res(v_R^{(1)};v)}{\lambda-v}-v_R^{(1)}(\lambda),&\mbox{ if $\lambda\in U_v$,}
\end{align}
and consequently we have the following asymptotics for $R_1$ defined by (\ref{RHP R:c}),
\begin{equation}\label{as R1}
R_1(x,t,\e)=\e^{\frac{1}{2}+|\Delta_k|}\Res(v_R^{(1)};v)+\bigO(\e\ln^2\e).
\end{equation}

Using the definition (\ref{def R imp}) of $R$, the improved outside parametrix given by (\ref{RHP Pinfty c2}), and the expansion (\ref{Pinftyimp}) of the outside parametrix at infinity, we obtain that
\begin{equation}
S_{11}(\lambda;x,t,\e)=1-i(R_{1,12}(x,t,\e)+C_{12}+2k\sqrt{v-u})\frac{1}{\sqrt{-\lambda}}+\bigO(\lambda^{-1}),\qquad \mbox{ as $\lambda\to\infty$,}
\end{equation}
which implies by (\ref{uS}) that
\begin{equation}\label{uR  imp1}
u(x,t,\epsilon)=u-2\epsilon\partial_x\left(
R_{1,12}(x,t,\e)+C_{12}(x,t,\epsilon)\right).
\end{equation}
Substituting the asymptotics obtained in (\ref{as R1}) gives
\begin{equation}\label{uR  imp2}
u(x,t,\epsilon)=u-2\epsilon\partial_x\left(
\e^{\frac{1}{2}+|\Delta_k|}\Res(v_R^{(1)};v)_{12}+C_{12}(x,t,\epsilon)\right)+\bigO(\e\ln^2\e).
\end{equation}
It is not so obvious that it is allowed to formally take derivatives of the asymptotics obtained for $R$.
However $R$ is analytic as a function of $y$, and the asymptotic expansion (\ref{Rexpansion}) can be shown to hold true for $y$ in a small complex neighborhood of any $y\in\mathbb R$. Using this property, (\ref{uR  imp2}) can be justified.

After some calculations we obtain
\begin{multline}\label{vR1}
2\e^{\frac{3}{2}}\partial_x [\e^{|\Delta_k|}\Res(v_R^{(1)};v)_{12}]\\=-8\e^{\frac{1}{2}+|\Delta_k|}(v-u)
\dfrac{|\Delta_k|}{\Delta_k}\left[-\dfrac{1-\delta}{2\pi h_k^2\gamma^{2k+1}}+2\pi \delta h_{k-1}^
2\gamma^{2k-1}\right],
\end{multline}
where $h_k$ is defined in (\ref{hk}) and $\gamma$ is defined in (\ref{gamma}).
By (\ref{def C imp1}),
\begin{multline}
\label{C1}2\e\partial_x C_{12}=-8(v-u)(1-\delta)\dfrac{\epsilon^{\frac{1}{2}+\Delta_k}2\pi h^2_{k-1}\gamma^{2k-1}}{\left( 1+\epsilon^{\frac{1}{2}+\Delta_k}2\pi h^2_{k-1}\gamma^{2k-1}\right)^2}\\
-8(v-u)\delta\dfrac{\epsilon^{\frac{1}{2}-\Delta_k}/(2\pi h^2_{k}\gamma^{2k+1})}{\left( 1+\epsilon^{\frac{1}{2}-\Delta_k}/(2\pi h^2_{k}\gamma^{2k+1})\right)^2}.
\end{multline}
Putting together the above formulas and using the fact that $\delta=1$ for $k\leq y\leq k+\frac{1}{2}$ and $\delta=0 $ for $k+\frac{1}{2}\leq y<k+1$ we arrive at the expression
\begin{equation}\label{uR imp3}
u(x,t,\epsilon)=u+2(v-u)\left[\sech^2(X_k)+4e^{-2X_{k-1}}\right]
+\bigO(\e\ln^2\e),
\end{equation}
if $k\leq y< k+1/2$, and at the expression
\begin{equation}\label{uR imp4}
u(x,t,\epsilon)=u+2(v-u)\left[\sech^2(X_k)+4e^{2X_{k+1}}\right]
+\bigO(\e\ln^2\e),
\end{equation}
if $k+1/2\leq y\leq k+1$,
where
\begin{equation}
X_k=\frac{1}{2}(\frac12-y+k)\ln\e-\ln(\sqrt{2\pi}
h_k)-(k+\frac12)\ln\gamma.
\end{equation}
We can write this also in the following more elegant form, which holds for all $|y|\leq M$:
\begin{equation}
u(x,t,\e)=u+2(v-u)\sum_{k=1}^{\infty}\sech^2(X_k)+\bigO(\e\ln^2\e).
\end{equation}
For each value of $y$, there are at most two of the $\sech^2$-terms which are bigger than $\bigO(\e\ln^2\e)$, all the others are absorbed by the error term.
This proves Theorem \ref{main theorem}.

\section*{Acknowledgements}
TC is a Postdoctoral Fellow of the Fund for Scientific Research - Flanders (Belgium), and was also supported by
Belgian Interuniversity Attraction Pole P06/02, and by the ESF program MISGAM.
TG  acknowledges support by the ESF program MISGAM.
TG and TC acknowledge support by ERC Advanced Grant FroMPDEs.

\obeylines \texttt{Tom Claeys
    Cit\'e Scientifique - Laboratoire Painlev\'e M2
    F-59655 Villeneuve d'Ascq, FRANCE
E-mail: tom.claeys@math.univ-lille1.fr
\bigskip

Tamara Grava
    SISSA Via Beirut 2-4
    34014 Trieste, ITALY
    E-mail: grava@sissa.it }

\end{document}